\documentclass[journals]{IEEEtran}

\usepackage[latin1]{inputenc}
\usepackage{graphicx, amssymb, amsmath, amsfonts, amsthm}
\usepackage[usenames,dvipsnames]{color}
\usepackage{bbm}

 \usepackage{hyperref}

\usepackage{algorithm}
\usepackage{algpseudocode}

\newtheorem{proposition}{Proposition}
\usepackage[short]{optidef}

\allowdisplaybreaks
\IEEEoverridecommandlockouts

\begin{document}
\title{MIMO with Analogue $1$-bit Phase Shifters: \\A Quantum Annealing Perspective}

\author{Ioannis Krikidis, \IEEEmembership{Fellow, IEEE}\vspace*{-5mm}
\thanks{I. Krikidis is with the Department of Electrical and Computer Engineering, University of Cyprus, Cyprus (email: krikidis@ucy.ac.cy).}}

\maketitle

\begin{abstract}
In this letter, we study the analogue pre/post-coding vector design for a point-to-point multiple-input multiple-output (MIMO) system with $1$-bit phase shifters. Specifically, we focus on the signal-to-noise ratio (SNR) maximization problem which corresponds to a combinatorial NP-hard due to the binary phase resolution. Two classical computation heuristics are proposed {\it i.e.,} i) an $1$-bit real-valued approximation of the optimal digital designs, and ii) an alternating optimization where a Rayleigh quotient problem is solved at each iteration. An iterative quantum annealing (QA)-based heuristic is also investigated, which outperforms classical counterparts and achieves near-optimal performance while ensuring polynomial time complexity. Experimental results in a real-world D-WAVE QA device validate the efficiency of the proposed QA approach. 
\end{abstract}
\vspace{-0.1cm}

\vspace{-0.2cm}
\section{Introduction}

\IEEEPARstart{l}{arge-scale} multiple-input multiple-output (MIMO) is an essential technology for the upcoming 6G communication systems towards extremely high spectral efficiency and reliability.  Nevertheless, conventional fully-digital implementations may be problematic, due to the increased hardware complexity and power consumption \cite{ALI}. To address this practical challenge, various MIMO architectures have been proposed in the literature to balance the trade-off between performance and implementation complexity. To name some of them, hybrid analogue-digital architectures, MIMO with low-resolution digital-to-analogue and/or analogue-to-digital converters, MIMO with analogue-only processing, have been proposed in the literature \cite{ALI,WAN,SIL}. Analogue-only MIMO architectures allow high-dimensional analogue pre/post-coding through phase shifters \cite[Sec. V. F]{ALI}, and are suitable for high frequency bands where the cost of radio frequency chains is a critical bottleneck. To further reduce complexity, low-resolution phase shifters ({\it e.g.,} $1$-bit) are used in practice \cite{WAN,SIL}; this discretization casts the pre/post-coding design into NP-hard combinatorial problems which can be solved through exhaustive search (ES) with exponential complexity.

A promising technology to solve NP-hard problems is to employ physics-inspired quantum computing techniques, which overcome classical computing barriers due the peculiar properties of quantum mechanics. Specifically, quantum-annealing (QA) is an analogue quantum model relying on the principles of the Adiabatic Theorem which enforces a quantum system to convergence to the ground state (lowest energy state) \cite{MCG} through Adiabatic evolution. By encoding the final Hamiltonian of the system to the desired problem, QA can be used to solve NP-hard combinatorial problems which are represented as quadratic unconstrained binary optimization (QUBO) instances.  Most of the relevant QA literature relies on the maximum-likelihood detection problem for large MIMO systems \cite{JAM1,JAM2}. Recent studies apply QA solvers in other wireless communications problems such as beamforming design in reconfigurable intelligent surfaces \cite{ROS}. The integration of QA techniques in wireless communications unlocks new degrees of freedom and enables the implementation of optimal solutions without computation restrictions.

\begin{figure}
	\centering
	\includegraphics[width=0.8\linewidth]{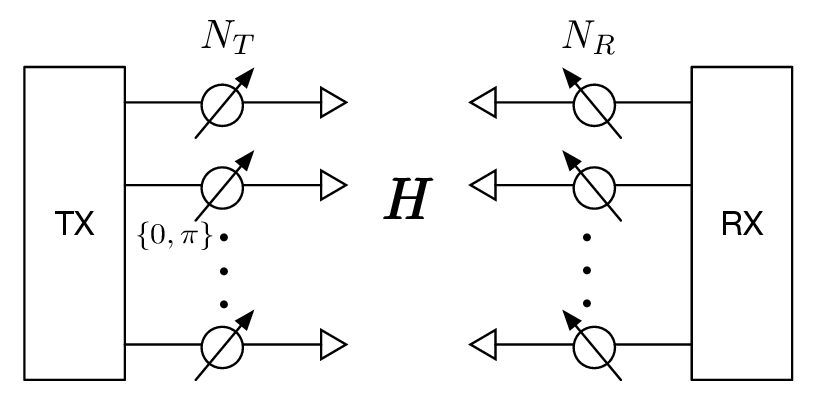}
	\vspace{-0.35cm}
	\caption{Point-to-point MIMO with $N_T$ transmit antennas, $N_R$ receive antennas and analogue $1$-bit pre/post-coding.}\label{fig_sys}
\end{figure}

In this work, we study the problem of analogue pre/post-coding design for a point-to-point MIMO with $1$-bit phase resolution \cite{WAN, SIL}; the analogue vectors are designed to maximize the received signal-to-noise ratio (SNR). The considered problem is combinatorial NP-hard and its optimal solution refers to ES with exponential complexity. Firstly, we introduce two classical-computing heuristics which discretize in a different fashion the solution of the equivalent continuous complex-valued problem {\it i.e.,} i) an $1$-bit real-valued approximation of the optimal digital solution based on the singular value decomposition (SVD) of the MIMO channel, and ii) an alternating optimization scheme that solves a complex-valued Rayleigh quotient problem at each iteration. Then, an iterative QA-based technique is proposed that takes into account the binary structure of the combinatorial design problem. The proposed algorithm is based on alternating optimization and solves appropriate QUBO problems at each iteration. Simulation and experimental results in a state-of-the-art quantum device (D-WAVE QA) \cite{WAVE} show that the quantum solution outperforms classical heuristics and achieves near-optimal performance (similar to ES), while ensuring polynomial time complexity.

{\it Notation:} Lower and upper case bold symbols denote vectors and matrices, respectively, the superscripts $(\cdot)^{\top}$, $(\cdot)^{H}$  denote
transpose and conjugate transpose, respectively, $\textrm{diag}(\pmb{x})$ is a 
diagonal matrix whose main diagonal is $\pmb{x}$, $\mathcal{CN}(\mu,\sigma^2)$ represents the complex Gaussian distribution with mean $\mu$ and variance $\sigma^2$, $\mathbbm{C}^{x\times y}$ denotes the space of $x\times y$
matrices with complex entries, $\mathcal{R}(\cdot)$ denotes the real part of its complex argument, $\mathbbm{1}$ denotes an all-ones column vector of appropriate dimension, and $\|\cdot\|_{\max}$ denotes the max-norm.

\vspace{-0.3cm}
\section{System model}

We consider a MIMO system consisting of $N_T$ and $N_R$ transmit and receive antennas, respectively. Let $\pmb{H} \in \mathbbm{C}^{N_R\times N_T}$ denote the MIMO channel matrix, where entries correspond to the channels between the transmit and receive antennas {\it i.e.}, $h_{i,j}$ is the channel coefficient between the $j$-th transmit antenna and the $i$-th receive antenna. We assume normalized Rayleigh block fading channels {\it i.e.,} $h_{i,j}\sim \mathcal{CN}(0,1)$. Fig. \ref{fig_sys} schematically presents the system model.

The MIMO system operates in the {\it beamforming mode} (single data flow) to maximize the received SNR and thus pre/post- processing vectors are adjusted accordingly \cite{HAM}. To reduce complexity and power consumption, both the transmitter and the receiver are equipped with $1$-bit resolution phase shifters\footnote{It serves as a useful guideline for more sophisticated schemes ({\it e.g.,} single or dual phase-shifters with higher resolution {\it etc.})}. Let $\pmb{f} \in \mathcal{S}^{N_T\times 1}$ and $\pmb{g} \in \mathcal{S}^{N_R\times 1}$ denote the (unormalized) pre-coding and post-coding vectors, respectively, taking values in the (spin) set $\mathcal{S}\triangleq \{-1,+1\}$ ($-1$ corresponds to a phase shift $\pi$ and $+1$ corresponds to a phase shift $0$); the normalized pre/post-coding vectors are given by $\pmb{f}/\|\pmb{f}\|=\pmb{f}/\sqrt{N_T}$ and $\pmb{g}/\|\pmb{g}\|=\pmb{g}/\sqrt{N_R}$, respectively. We assume a perfect channel state information at both the transmitter and the receiver. The received SNR is given by
\vspace{-0.2cm}
\begin{align}
\rho(\pmb{g},\pmb{f})=\frac{P|\pmb{g}^T\pmb{H}\pmb{f}|^2}{N_T N_R \sigma^2},\label{snr_ex}
\end{align}
where $\sigma^2$ is the variance of the additive white Gaussian noise, $P$ is the transmit power, while the terms $N_T$ and $N_R$ in the denominator are due to the power normalization of the pre/post-coding vectors. Since the objective of the MIMO system is to maximize the received SNR, we introduce the following design problem 
\vspace{-0.2cm}
\begin{align}
&\max_{\pmb{f} \in \mathcal{S}^{N_T\times 1},\; \pmb{g} \in \mathcal{S}^{N_R \times 1}}\; |\pmb{g}^T\pmb{H}\pmb{f}|^2. \label{norm}
\end{align}
Due to the binary nature of the analogue pre/post-coding vectors, the above optimization problem is combinatorial NP-hard; the optimal solution requires ES over all the possible pre/post-coding vectors. 

\vspace{-0.6cm}
\subsection{Exhaustive search (ES)- Benchmark}\label{es1}

The ES scheme evaluates the SNR expression in \eqref{snr_ex} for all the possible transmit/receive vectors and returns the solution with the maximum SNR. The ES requires $2^{N_T+N_R}$ computations and therefore its complexity becomes exponential with the number of transmit/receive antennas. A more intelligent ES scheme can take into account the computation symmetry {\it i.e.,} $\rho(\pmb{g},\pmb{f})=\rho(\pm\pmb{g}, \pm \pmb{f})$ and therefore the number of ES computations decreases to $2^{N_T+N_R-2}$; for both ES schemes, the implementation is prohibited for large-scale MIMO topologies.

\vspace{-0.4cm}
\section{Design $1$-bit analogue pre/post-coding vectors}

In this section, the proposed classical and quantum computing heuristics are presented.

\vspace{-0.5cm}
\subsection{SVD-based design}  

This algorithm is inspired by the optimal (digital) solution for the considered problem which is the single-mode eigenbeamforming \cite[Sec. 3.6]{HAM}; the optimal unquantized solution is given by the first right/left eigenvectors of the SVD of the matrix $\pmb{H}$. More specifically, let $\pmb{H}=\pmb{U}\pmb{\Lambda}^{\frac{1}{2}}\pmb{V}^H$ denote the SVD decomposition \cite{STR} of the MIMO channel matrix $\pmb{H}$, where $\pmb{\Lambda}=\textrm{diag}(\lambda_1, \ldots, \lambda_{M})$ is a diagonal matrix with $\lambda_1\geq \lambda_2\geq\cdots \geq \lambda_M$, $\mathbf{U}=[\pmb{u}_1, \pmb{u}_2, \ldots, \pmb{u}_{M}] \in \mathbbm{C}^{N_R\times M}$ and $\pmb{V}=[\pmb{v}_1, \pmb{v}_2, \ldots, \pmb{v}_{M}] \in \mathbbm{C}^{N_T\times M}$ are both complex unitary matrices of
appropriate dimension, and $M=\min(N_T,N_R)$. Then, the optimal digital pre/post-coding vectors are equal to $\pmb{v}_1$ and $\pmb{u}_1$, respectively, corresponding to a maximum SNR $\rho=P\lambda_1/(N_T N_R \sigma^2)$. The proposed algorithm approximates the optimal digital pre/post-coding vectors with the $1$-bit real-valued vectors that minimize the mean square error (MSE). Specifically, we introduce the following (decoupled) optimization problems  
\begin{align}
&\min_{\pmb{f} \in \mathcal{S}^{N_T\times 1}}\|\pmb{v}_1-\pmb{f}\|^2,\;\;\;\;\min_{\pmb{g} \in \mathcal{S}^{N_R\times 1}}\|\pmb{u}_1-\pmb{g}\|^2.\label{obj1}
\end{align}
The solution to the above optimization problems is given by the following proposition. 
\begin{proposition}
The $1$-bit analogue pre/post-coding vectors that minimize the MSE expressions in \eqref{obj1} are given by
\begin{align}
\pmb{f}^{\textrm{svd}}=\textrm{sign}(\mathcal{R}(\pmb{v}_1)),\;\;\pmb{g}^\textrm{svd}=\textrm{sign}(\mathcal{R}(\pmb{u}_1)),
\end{align}
where the operator $\textrm{sign}(\cdot)$ returns the sign of its argument. 
\end{proposition}

\begin{proof}
We show the proof for the pre-coding vector (the analysis is similar for the post-coding vector). The first objective function in \eqref{obj1} can be written as
\begin{align}
\|\pmb{v}_1-\pmb{f}\|^2&=\|\pmb{v}_1\|^2+\|\pmb{f}\|^2-2\pmb{f}^T\mathcal{R}(\pmb{v}_1) \nonumber \\
&=1+N_T-2\pmb{f}^T\mathcal{R}(\pmb{v}_1), \label{obj2}
\end{align}
with $\pmb{f}^T\mathcal{R}(\pmb{v}_1)=\sum_{i=1}^{N_T}f_{i}\mathcal{R}(\pmb{v}_{1,i})\leq \sum_{i=1}^{N_T}|f_{i}\mathcal{R}(\pmb{v}_{1,i})|$, where equality holds for $f_{i}\mathcal{R}(\pmb{v}_{1,i})\geq 0 \Rightarrow \textrm{sign}(f_i)=\textrm{sign}(\mathcal{R}(\pmb{v}_{1,i})$; therefore $\pmb{f}^{\textrm{svd}}=\arg \min_{\pmb{f} \in \mathcal{S}^{N_T\times 1}}\|\pmb{v}_1-\pmb{f}\|^2=\arg \min_{\pmb{f} \in \mathcal{S}^{N_T\times 1}} -\pmb{f}^T\mathcal{R}(\pmb{v}_1)= \textrm{sign}(\mathcal{R}(\pmb{v}_1)),$
which completes the proof.
\end{proof}
The complexity of this scheme is mainly related to the computation of the SVD of the channel matrix \cite[II.1]{STR}; despite its low computation complexity, this heuristic suffers from high quantization distortion (due to the $1$-bit real-valued approximation of the optimal complex-valued SVD solution) and therefore its performance is highly suboptimal.

\vspace{-0.3cm}
\subsection{Rayleigh quotient-based design}

We propose an iterative technique (alternating optimization) to design the analogue pre/post-coding vectors by using the maximization of the Rayleigh quotient \cite[I.10]{STR}. More specifically, the objective function in \eqref{norm} can be written as follows
\begin{align}
|\pmb{g}^T\pmb{H}\pmb{f}|^2=\pmb{g}^T(\pmb{H}\pmb{f}\pmb{f}^T\pmb{H}^H)\pmb{g}=\pmb{f}^T(\pmb{H}^H\pmb{g}\pmb{g}^T\pmb{H})\pmb{f}.\label{exp0}
\end{align}
Since the matrices in the middle of the above expressions are positive-semidefinite, we  propose an alternating optimization technique by relaxing the spin (binary) vectors. Specifically, given a fixed post-coding vector $\pmb{g}$, the SNR maximization problem can be written as a Rayleigh quotient with 
\begin{align}
\pmb{f}_r^*=\arg \max_{\pmb{f}_r} \frac{\pmb{f}_r^T(\pmb{H}^H\pmb{g}\pmb{g}^T\pmb{H})\pmb{f}_r}{\pmb{f}_r^T\pmb{f}_r}=\pmb{z}_1(\pmb{H}^H\pmb{g}\pmb{g}^T\pmb{H}),
\end{align}   
where $\pmb{f}_r$ denotes the relaxed (complex) pre-coding vector with $\|\pmb{f}_r\|=1$, and $\pmb{z}_1(\pmb{A})$ returns the eigenvector of the matrix $\pmb{A}$ corresponding to the maximum eigenvalue \cite{STR}. The binary pre-coding vector is then given by $\pmb{f}=\textrm{sign}(\mathcal{R}(\pmb{f}_r^*))$ through a $1$-bit approximation of the real part.  In the second step of the algorithm, we fix the pre-coding vector $\pmb{f}$ (by using the previous solution) and we optimize the SNR expression with the respect to the post-coding vector $\pmb{g}$ by using a similar procedure; the optimal relaxed post-coding vector is given by the following Rayleigh quotient
\begin{align}
\pmb{g}_r^*=\arg \max_{\pmb{g}_r} \frac{\pmb{g}_r^T(\pmb{H}\pmb{f}\pmb{f}^T\pmb{H}^H)\pmb{g}_r}{\pmb{g}_r^T\pmb{g}_r}=\pmb{z}_1 (\pmb{H}\pmb{f}\pmb{f}^T\pmb{H}^H),
\end{align} 
which gives $\pmb{g}=\textrm{sign}(\mathcal{R}(\pmb{g}_r^*))$. This process is repeated until convergence or a maximum number of iterations $K$ is achieved. A simple modification of this algorithm is to keep the relaxed complex vectors $\pmb{f}_r$, $\pmb{g}_r$ until convergence and truncate into $1$-bit real-valued vectors at the end; both heuristics are considered in our numerical studies. The pseudocode of the two Rayleigh quotient-based algorithms is presented in Algorithm \ref{algo1} (RQ) and Algorithm \ref{algo2} (RQ-M), respectively. It is worth noting that both RQ heuristics discretize the solution of the relaxed (complex-valued) problem; although the relaxed problem converges to global/local digital solution, there is not any converge guarantee for the considered RQ heuristics.

\begin{algorithm}[t]
\caption{Rayleigh quotient pre/post-coding}\label{algo1}\vspace{1mm}
\hspace*{\algorithmicindent}\textbf{Input:} $\pmb{H}$, initial vectors $\pmb{g}^{(0)} \in \mathcal{S}^{N_R}$, $\pmb{f}^{(0)} \in \mathcal{S}^{N_T}$, \hspace*{\algorithmicindent}relative tolerance $\delta$, $K$, $\rho_{\textrm{old}}=\rho(\pmb{g}^{(0)},\pmb{f}^{(0)})$, $k\leftarrow 0$.
\begin{algorithmic}[1]
\Repeat
  \State $k \leftarrow k + 1$
  \If{$k>1$}
    \State Let $\rho_{\textrm{old}}=\rho_{\textrm{new}}$
  \EndIf
  \State Compute $\pmb{f}_r^{(k)}=\pmb{z}_1(\pmb{H}^H\pmb{g}^{(k-1)}{\pmb{g}^{(k-1)}}^T\pmb{H})$.
  \State Obtain $\pmb{f}^{(k)}=\textrm{sign}(\mathcal{R}(\pmb{f}_r^{(k)}))$.
  \State Compute $\pmb{g}_r^{(k)}=\pmb{z}_1 (\pmb{H}\pmb{f}^{(k)}{\pmb{f}^{(k)}}^T\pmb{H}^H)$.
  \State Obtain $\pmb{g}^{(k)}=\textrm{sign}(\mathcal{R}(\pmb{g}_r^{(k)}))$.
  \State Let $\rho_{\textrm{new}}=\rho(\pmb{g}^{(k)},\pmb{f}^{(k)})$.
\Until{$|\rho_{\textrm{new}}-\rho_{\textrm{old}}|/|\rho_{\textrm{old}}|<\delta$ or $k\geq K$}
\end{algorithmic}
	\hspace*{\algorithmicindent} \textbf{Output:} $(\pmb{g}^{\textrm{rq}},\pmb{f}^{\textrm{rq}})=(\pmb{g}^{(k)},\pmb{f}^{(k)})$.
\end{algorithm}

\begin{algorithm}[t]
\caption{Rayleigh quotient (modified) pre/post-coding}\label{algo2}\vspace{1mm}
\hspace*{\algorithmicindent}\textbf{Input:} $\pmb{H}$, initial coding vectors $\pmb{g}_r^{(0)} \in \mathbbm{C}^{N_R}$, $\pmb{f}_r^{(0)} \in \mathbbm{C}^{N_T}$ \hspace*{\algorithmicindent}relative tolerance $\delta$, $K$, $\rho_{\textrm{old}}=\rho(\pmb{g}_r^{(0)},\pmb{f}_r^{(0)})$, $k\leftarrow0$.
\begin{algorithmic}[1]
\Repeat
  \State $k \leftarrow k + 1$
  \If{$k>1$}
    \State Let $\rho_{\textrm{old}}=\rho_{\textrm{new}}$
  \EndIf
  \State Compute $\pmb{f}_r^{(k)}=\pmb{z}_1(\pmb{H}^H\pmb{g}_r^{(k-1)}{\pmb{g}_r^{(k-1)}}^T\pmb{H})$.
  \State Compute $\pmb{g}_r^{(k)}=\pmb{z}_1 (\pmb{H}\pmb{f}_r^{(k)}{\pmb{f}_r^{(k)}}^T\pmb{H}^H)$.
  \State Let $\rho_{\textrm{new}}=\rho(\pmb{g}_r^{(k)},\pmb{f}_r^{(k)})$
\Until{$|\rho_{\textrm{new}}-\rho_{\textrm{old}}|/|\rho_{\textrm{old}}|<\delta$ or $k\geq K$}
\end{algorithmic}
	\hspace*{\algorithmicindent} \textbf{Output:} $(\pmb{g}^{\textrm{rqm}},\pmb{f}^{\textrm{rqm}})=(\textrm{sign}(\mathcal{R}(\pmb{g}_r^{(k)})),\textrm{sign}(\mathcal{R}(\pmb{f}_r^{(k)})))$.
\end{algorithm}

\vspace{-0.3cm}
\subsection{QA-based design}

Firstly, we provide the basic background of the QA and then we present the proposed iterative algorithm and the associated QUBO formulations. 

\subsubsection{QA background}\label{back}
Adiabatic quantum computing is a promising tool to solve NP-hard combinatorial problems by using the principles of quantum mechanics. Specifically, according to the Adiabatic Theorem \cite{MCG}, if a quantum system is initially in the ground state (quantum state with the lowest energy) of an initial Hamiltonian and the system evolves/changes slowly, it will converge to the ground state of the final Hamiltonian. By encoding the final Hamiltonian to the desired optimization problem, the Adiabatic evolution can be used to solve complex combinatorial problems which are represented as QUBO instances. D-WAVE is a commercial analogue quantum computer/device that implements a noisy approximation of the quantum Adiabatic evolution called QA \cite{MCG,KAS}. This quantum device has received considerable interest lately due to the high number of available qubits (more than $5,000$ qubits in latest hardware architectures, which are arranged in a specific hardware topology/graph) and its friendly interface for remote access (D-WAVE Leap) \cite{WAVE}.

The process of mapping a QUBO problem (which represents the quadratic interconnection of the binary variables) into the limited D-WAVE QA hardware topology/graph, it is called {\it minor embedding}; since the hardware graph is not fully connected, this process enables logical channelling between qubits where logical variables are represented by chains of multiple physical qubits. The logical channelling is characterized by the strength of the logical links (called {\it ferromagnetic} coupling) and it is a critical parameter for the QA performance. In case that a logical chain is broken at the end of QA process ({\it i.e.,} physical qubits that form a logical qubit have different final values), an appropriate consensus algorithm is applied. 

In this work, we consider the heuristic algorithm {\it minorminer} for minor embedding which is included in the Ocean software development kit by default; in this case, a majority vote is applied to broken chains \cite{WAVE}. Due to practical non-idealities ({\it e.g.,} hardware limitations, Hamiltonian noise, temperature fluctuations {\it etc.}), the output of a single D-WAVE run (referred to as an anneal) is probabilistic and may be different than the ground state under question. To ensure efficient solution for the considered QUBO problem, it is a common practice to solve the same QUBO instance multiple times; the best solution among all the anneals, it is the final D-WAVE QA output. The anneal time and the number of anneals are critical design parameters which are tuned empirically.  It is worth noting that in QA, instead of analysing the computation time/complexity of a given algorithm, we mainly study the trade-off between the time and the probability that the QA output is correct.

\begin{algorithm}[t]
	\caption{QA-based pre/post-coding}\label{algo3}\vspace{1mm}
	\hspace*{\algorithmicindent}\textbf{Input:} $\pmb{H}$, $\pmb{g}_l^{(0)} \in \mathcal{S}^{N_R}$, $\pmb{f}_l^{(0)} \in \mathcal{S}^{N_T}$ with $l=1,\dots,L$, \hspace*{\algorithmicindent}relative tolerance $\delta$, $L$, $K$, $\rho_{\textrm{old}_l}= \rho(\pmb{g}_l^{(0)},\pmb{f}_l^{(0)})$, $k\leftarrow 0$.
	\begin{algorithmic}[1]
		\For{$l=1, 2,\ldots, L$} 
		\Repeat
		\State $k \leftarrow k + 1$
		\If{$k>1$}
		\State Let $\rho_{\textrm{old}_l}=\rho_{\textrm{new}_l}$
		\EndIf
		\State Compute $\pmb{Q}\!=\! \pmb{H}^H\pmb{g}_l^{(k-1)}{\pmb{g}_l^{(k-1)}}^T\pmb{H}$; convert to $\pmb{Q}_n$.
		\State \textbf{[D-WAVE]} Solve $\pmb{b}_f^*=\arg \min_{\pmb{b}_f}\pmb{b}_f^T(-\pmb{Q}_n)\pmb{b}_f$.
		\State Convert $\pmb{b}_f^*$ to spin vector $\pmb{f}_l^{(k)}$. 
		\State Compute $\pmb{R}= \pmb{H}\pmb{f}_l^{(k)}{\pmb{f}_l^{(k)}}^T\pmb{H}^H$ and convert to $\pmb{R}_n$.
		\State \textbf{[D-WAVE]} Solve $\pmb{b}_g^*=\arg \min_{\pmb{b}_g}\pmb{b}_g^T(-\pmb{R}_n)\pmb{b}_g$.
		\State Convert $\pmb{b}_g^*$ to spin vector $\pmb{g}_l^{(k)}$. 
		\State Let $\rho_{\textrm{new}_l}=\rho(\pmb{g}_l^{(k)},\pmb{f}_l^{(k)})$.
		\Until{$|\rho_{\textrm{new}_l}-\rho_{\textrm{old}_l}|/|\rho_{\textrm{old}_l}|<\delta$ or $k\geq K$}
		\State Obtain $(\pmb{g}_l,\pmb{f}_l)=(\pmb{g}_l^{(k)},\pmb{f}_l^{(k)})$.
		\EndFor
			\end{algorithmic}
		\hspace*{\algorithmicindent} \textbf{Output:} $(\pmb{g}^{\textrm{qa}},\pmb{f}^{\textrm{qa}})=\arg \max_{\pmb{g}_l,\pmb{f}_l} \rho(\pmb{g}_l,\pmb{f}_l)$.
\end{algorithm}

\subsubsection{An iterative QA algorithm}

We propose a new algorithm that incorporates a QA solver in the design of the analogue pre/post-coding vectors. By taking into account the binary combinatorial structure of the problem, we introduce an alternating optimization algorithm (similar to the Rayleigh quotient-based designs) that solves a QUBO problem at each iteration. The QUBO formulations are solved through a state-of-the-art D-WAVE QA device. More specifically, based on the expressions in \eqref{exp0}, we firstly fix the vector $\pmb{g}$ (by using a random solution) and we study the following combinatorial problem with respect to the pre-coding vector $\pmb{f}$, given by
\begin{align}
\max_{\pmb{f}} \pmb{f}^T(\pmb{H}^H\pmb{g}\pmb{g}^T\pmb{H})\pmb{f}=\max_{\pmb{f}} \pmb{f}^T\pmb{Q}\pmb{f}. \label{is1}
\end{align} 
We convert the pre-coding vector $\pmb{f}$ (containing spin variables $\{+1,-1\}$) to the binary vector $\pmb{b}_f$ (with entries in the set $\{0,1\}$) by using the transformation $\pmb{b}_f=\frac{1}{2}(\pmb{f}+\mathbbm{1})$ and thus the initial problem in \eqref{is1} is converted to a QUBO formulation
\begin{align}
\max_{\pmb{f}} \pmb{f}^T\pmb{Q}\pmb{f}\;=&\max_{\pmb{b}_f}\;4 \pmb{b}_f^T \pmb{Q}\pmb{b}-2\pmb{b}_f^T\pmb{Q}\mathbbm{1}-2\mathbbm{1}^T\pmb{Q}\pmb{b}_f \nonumber \\
%&=\max_{\pmb{b}_f}\; \pmb{b}_f^T \bigg(4\pmb{Q}-4\textrm{diag}%(\mathcal{R}\big(\pmb{Q}\mathbbm{1})\big) \bigg)\pmb{b}_f \nonumber \\
&=\max_{\pmb{b}_f}\; \pmb{b}_f^T \mathcal{R} \bigg(4\pmb{Q}-4\textrm{diag}(\mathcal{R}\big(\pmb{Q}\mathbbm{1})\big) \bigg)\pmb{b}_f \nonumber \\
&=\max_{\pmb{b}_f}\; \pmb{b}_f^T \pmb{Q}_0 \pmb{b}_f.
\end{align}
The last mathematical step transforms the above (symmetric) quadratic matrix $\pmb{Q}_0$ to a form that is compatible with the D-WAVE QA solver \cite{JAM2}; we normalize the matrix such as all its entries take values in the range $[-1,+1]$. The calibrated matrix is written by $\pmb{Q}_n=\frac{\pmb{Q}_0}{\| \pmb{Q}_0\|_{\max}}$ and thus the D-WAVE compatible QUBO problem is written as
\begin{align}
\min_{\pmb{b}_f}\;\pmb{b}_f^T (-\pmb{Q}_n) \pmb{b}_f.
\end{align}
The returned binary vector $\pmb{b}_f$ is converted back to the spin vector $\pmb{f}$. By using the principles of alternating optimization, we then fix $\pmb{f}$ (by using the previous solution), and we solve the following QUBO problem (the transformation follows similar analytical steps)
\begin{align}
\max_{\pmb{g}}\; \pmb{g}^T(\pmb{H}\pmb{f}\pmb{f}^T\pmb{H}^H)\pmb{g}&=\max_{\pmb{g}}\; \pmb{g}^T\pmb{R}\pmb{g} \nonumber \\
&\Rightarrow \min_{\pmb{b}_g}\; \pmb{b}_g^T(-\pmb{R}_n)\pmb{b}_g,
\end{align}
where $\pmb{b}_g$ is the binary representation of the vector $\pmb{g}$. The above alternating optimization process is repeated until convergence or a maximum number of iterations $K$ is achieved. Due to the discrete/binary nature of the problem, the final solution is sensitive to the initial conditions ({\it i.e.,} initial vector $\pmb{g}^{(0)})$ which could result in convergence to a local maximum. To overcome this limitation and facilitate the iterative algorithm to escape local maxima and get closer to the optimal solution, we introduce $L$ independent execution of the algorithm by using different initial conditions ($\pmb{g}_l^{(0)}$ with $l=1,\ldots,L$); from all the returned solutions, we keep the one corresponding to the maximum SNR objective function. The pseudocode of the QA-based iterative algorithm is given in Algorithm \ref{algo3}.     

\vspace{-0.1cm}
\section{Evaluation}

Computer simulations and experimental (D-WAVE) results are carried-out to evaluate the performance of the proposed schemes. We consider two basic MIMO scenarios with $N_T=N_R=\{8,10\}$ and system parameters $\sigma^2=1$, $P=0$ dB, $K=10$, $L=10$, $\delta=0.01$. For our D-WAVE QA experiments, we use the D-WAVE Leap interface with the {\it Advantage\_system1.1} quantum processing unit \cite{WAVE} with $1,000$ anneals, $1$ $\mu\textrm{sec}$ annealing time, and a ferromagnetic coupling parameter $\overline{J}_F=3$; these parameters are tuned empirically. For the minor embedding process, we follow the discussion in Sec. \ref{back}. Due to the probabilistic nature of the QA process, each anneal returns distinct solutions; the best solution among all the anneals (corresponding to maximum SNR) is selected as the final D-WAVE QA output. The single-mode eigenbeamforming \cite{HAM} is considered as a performance benchmark/bound.

\begin{figure}
\centering
\includegraphics[width=0.8\linewidth]{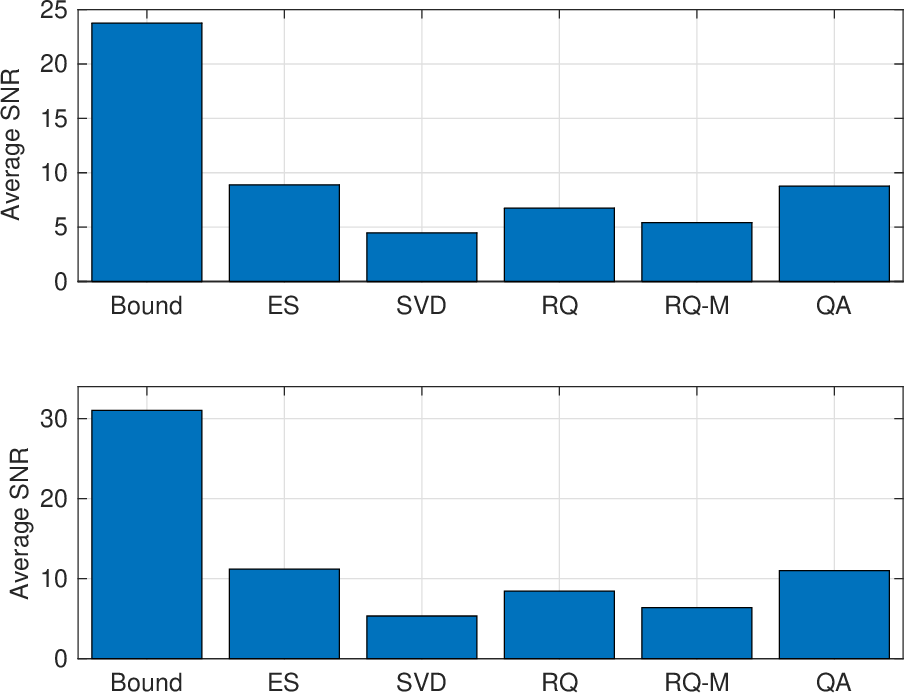}
\vspace{-0.3cm}
\caption{Average SNR performance for the proposed vector designs; [top] MIMO with $N_T=N_R=8$, [bottom] MIMO with $N_T=N_R=10$.} \label{figres1}
\end{figure}

\begin{figure}
\centering
	\includegraphics[width=0.82\linewidth]{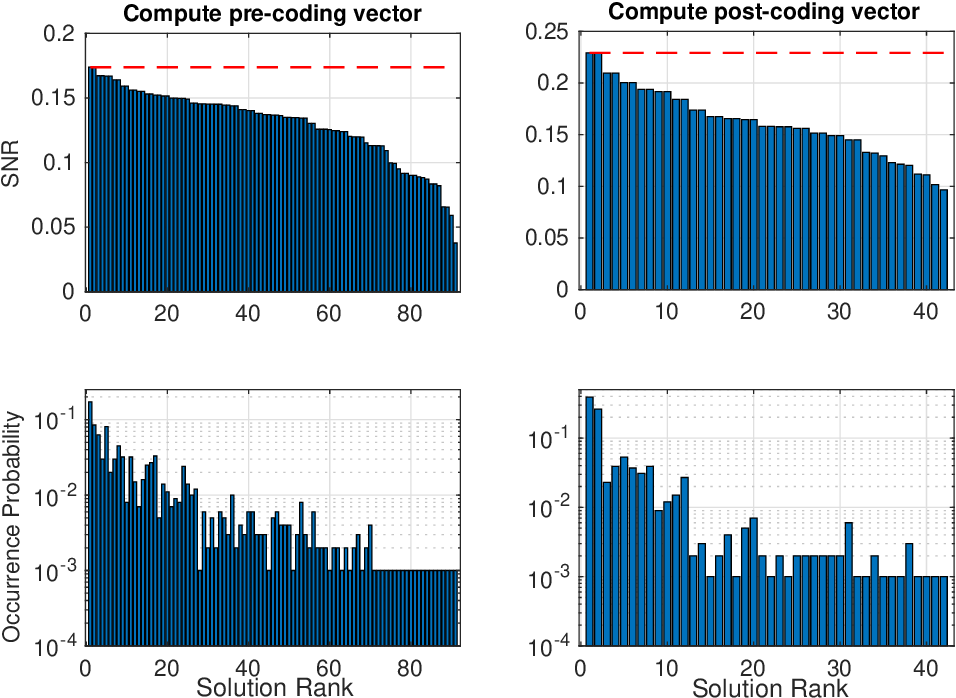}
	\vspace{-0.35cm}
	\caption{D-WAVE performance for a single iteration and channel realization with $N_T=N_R=8$; ES benchmark (dashed line). [top] SNR performance of the returned solutions in descending order, [bottom] Occurrence probability of the returned solutions.}\label{figres2}
\end{figure}

\begin{figure}
\centering
	\includegraphics[width=0.82\linewidth]{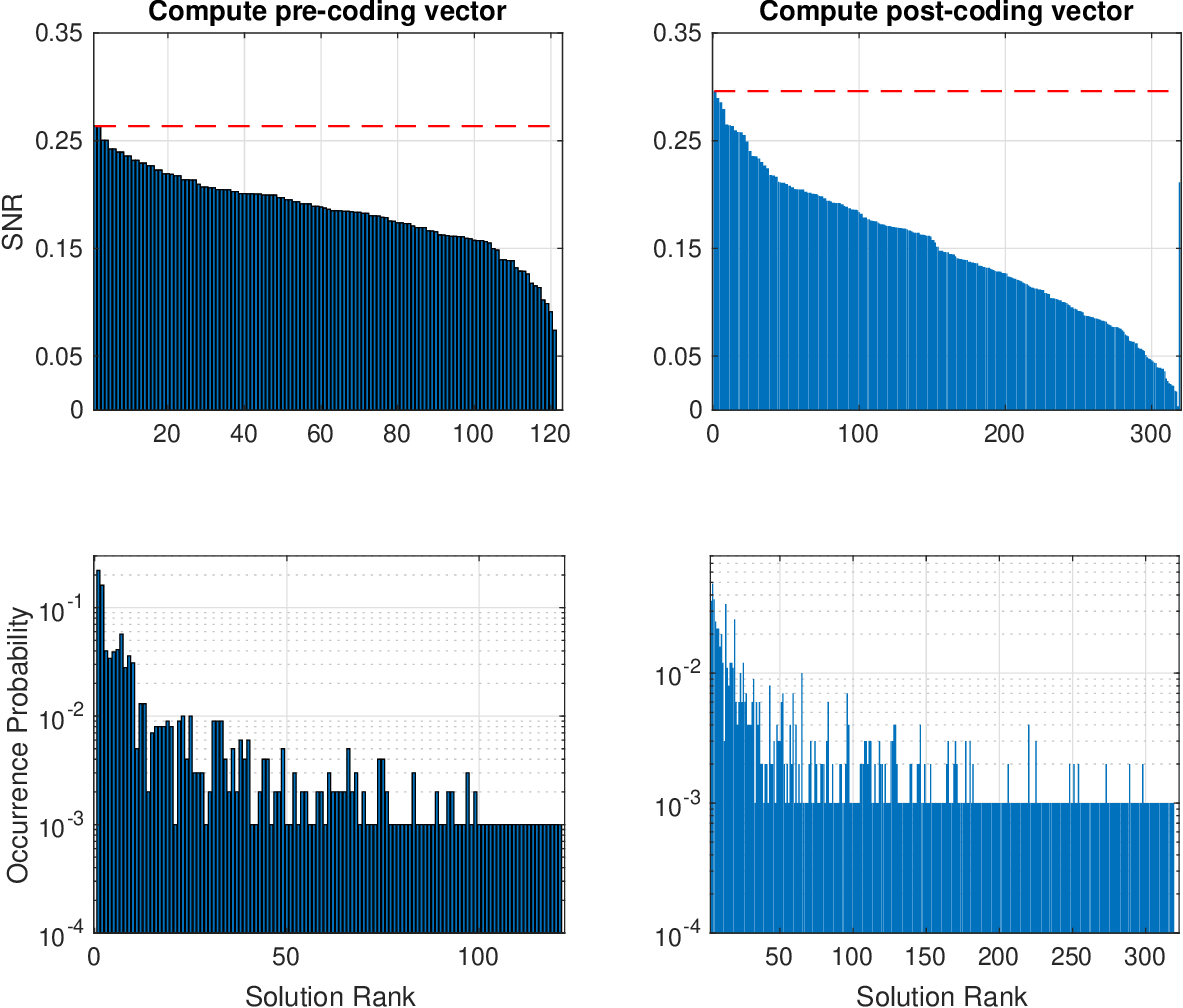}
	\vspace{-0.35cm}
	\caption{D-WAVE performance for a single iteration and channel realization with $N_T=N_R=10$; ES benchmark (dashed line). [top] SNR performance of the returned solutions in descending order, [bottom] Occurrence probability of the returned solutions.}\label{figres3}
\end{figure}

In Fig. \ref{figres1}, we compare the performance of the proposed schemes in terms of the average achieved SNR over $1,000$ channel realizations. We observe that the QA solution achieves a near-optimal performance (similar to ES) while ensuring polynomial processing time. Classical heuristics (SVD, RQ, RQ-M) discretize in a different way the complex-valued solution of the associated continuous problem and are suboptimal in comparison to ES and QA; among them, the RQ algorithm outperforms RQ-M and SVD approaches. Classical heuristics sacrifice performance but are appropriate for scenarios with extremely low-complexity requirements; the QA-based algorithm combines near-optimal performance with polynomial time processing and  is promising for future massive MIMO systems with quantum computing capabilities.

In Fig. \ref{figres2}, we focus on the performance of the D-WAVE QA solver in one basic algorithmic iteration for a single channel realization. Specifically, we show the returned solutions over $1,000$ anneals which are ordered in descending order of their SNR values as well as the associated occurrence probabilities. The sub-figures on the left refer to the pre-coding design $\pmb{f}_l^{(k)}$ (given a post-coding vector $\pmb{g}_l^{(k-1)}$); it can be seen that the best returned solution achieves the optimal ES performance and occurs with a probability $\approx 0.2$. It is worth noting that the first two (in the order) solutions are equivalent in terms of SNR performance, due to computation symmetry of the problem (see Sec. \ref{es1}); therefore, the total occurrence probability of the optimal solution becomes $\approx 0.3$. The sub-figures on the right refers to the next algorithmic step associated with the post-coding vector design $\pmb{g}_l^{(k)}$ (given the pre-coding vector from the previous step). We observe that the best two returned solutions are also equivalent and achieve the ES performance, while the final SNR value is slightly improved in comparison to the previous step of the algorithm. Similar observations are obtained in Fig. \ref{figres3} for a more complex MIMO setup with $N_T=N_R=10$.

To demonstrate the time performance of the D-WAVE QA, in Table \ref{tab1}, we show time results associated with the two QUBO problems in Figures \ref{figres2} and \ref{figres3}, respectively. The entire latency \cite{KAS} consist of {\it programming time} (pro-processing time to load the QUBO weights), {\it anneal time} (actual implementation of the QA, {\it i.e.,} $1$ msec since $1,000$ anneals in total), {\it readout time} (time to read the result at each anneal), {\it readout delay} (time to reset the qubits between anneals), and {\it post-processing time} (time to process the returned solutions). It can be seen that the real quantum processing time (anneal time) which is a design parameter that is tuned empirically, corresponds to a small fraction of the total processing time. The non-quantum time latency ({\it i.e.,} programming time, readout time/delay) is due to the conventional analogue/digital circuits that support the quantum device and seems to be (currently) the main bottleneck for delay-sensitive applications; it is technology related and is expected to be significantly lower in the near future \cite{KAS}.

\begin{table}
\centering
\caption{Indicative timing results for $1,000$ anneals ($\mu$sec)}
\vspace{-0.3cm}
\resizebox{\linewidth}{!}{
\begin{tabular}{|l||l|l|}
\hline
{\bf Time} & {\bf Case I ($N_T=N_R=8$)} & {\bf Case II ($N_T=N_R=10$)} \\
\hline 
Programming time & $15,762$ & $15,761$ \\
\hline 
Anneal time & $1,000$ & $1,000$ \\
\hline 
Readout time & $57,380$ & $138,180$ \\
\hline 
Readout delay & $20,540$ & $20,540$\\
\hline 
Post processing & $875$ & $243$ \\
\hline 
QPU Access time & $94,682$ & $175,481$ \\
\hline
\end{tabular}}\label{tab1}
\end{table}

\end{document}